\documentclass[12pt]{article} 
\usepackage{authblk}
\usepackage{amsmath,amsthm,amssymb}
\usepackage{graphicx,float}
\usepackage{xcolor}
\usepackage{xargs}
\usepackage{algpseudocode}
\usepackage{multirow}
\usepackage{multicol,ulem}

\usepackage{tikz}
\usetikzlibrary{decorations}

\newtheorem{theorem}{Theorem}[section]
\newtheorem{prop}[theorem]{Proposition}

\newtheorem{definition}[theorem]{Definition}
\newtheorem{remark}[theorem]{Remark}

\newcommand{\w}{\mathbf{w}}

\newcommand{\D}{\mathcal{D}}
\newcommand{\CH}{\mathcal{C}(H)}

 \usepackage[colorinlistoftodos,prependcaption,textsize=tiny]{todonotes} 
 
\newcommandx{\mm}[2][1=]{\todo[linecolor=red,backgroundcolor=yellow!20,#1]{#2 \\ \hfill --- Martín}}

\newcommandx{\lmcom}[2][1=]{\todo[linecolor=red,backgroundcolor=green!20,#1]{#2 \\ \hfill --- Luis}}

\title{The Minimum Clique Routing Problem on Cycles\footnote{Corresponding author: ptolomei@fceia.unr.edu.ar}}
\author[1]{M. Escalante and P. Tolomei}
\author[2]{M. Matamala and I. Rapaport}
\author[3]{L.M. Torres}

\affil[1]{\small FCEIA, Universidad Nacional de Rosario and CONICET, Argentina}
\affil[2]{\small DIM-CMM (ILR-CNRS-2807), Universidad de Chile, Chile}
\affil[3]{\small MODEMAT, 
Escuela Polit{\'e}cnica Nacional, Ecuador}

\date{}

\begin{document}

\maketitle

\normalem

\begin{abstract}
\noindent In the Minimum Clique Routing Problem on Cycles \textsc{MCRPC} we are given a cycle together with a set of demands (weighted origin-destination pairs) and the goal is to route all the pairs minimizing the maximum weighted clique of the intersection graph induced by the routing.
The vertices of this graph are the demands with their corresponding weights and two demands are adjacent when their routes share at least one arc.
In this work we are not only interested in the \textsc{MCRPC} but also in two natural subproblems. First, we consider the situation where the demands are disjoint, in the sense that every two demands do not share any of their corresponding ends. Second, we analyze the subproblem where the weights of the routes are all equal.   
We first show that the  problem is NP-complete even in the subproblem  of disjoint demands. 
For the case of arbitrary weights, we exhibit a simple combinatorial 2-approximation algorithm and a $\frac{3}{2}$-approximation algorithm based on rounding a solution of a relaxation of an integer linear programming formulation of our problem. 
Finally, we give a Fixed Parameter Tractable algorithm for the case  of uniform weights, whose parameter is related to the maximum degree of the intersection graph induced by any routing.
   
\end{abstract}

\section{Introduction}

In the Minimum Clique Routing Problem \textsc{MCRP} \cite{stefanakos2004routing}, the input is an undirected graph and a set of weighted origin-destination pairs of nodes $\D$.
The goal is to 
specify undirected paths in this graph (routes) connecting each origin node with the corresponding 
destination node in such a way that the maximum weighted clique in the intersection graph induced by the routing is minimized, where two routes intersect if they have at least one common arc and the weight of a demand is associated to the corresponding route.

The \textsc{MCRP} is NP-hard in general, as deciding whether there is a disjoint set of routes connecting all pairs in $\D$ corresponds to the Edge Disjoint Path Problem \textsc{EDPP}, a classical NP-complete problem \cite{Karp}, which remains NP-complete even when restricted to planar graphs \cite{Lynch} and to series-parallel graphs \cite{NishizekiEtAl}. On the other hand,  the \textsc{EDPP} admits a Fixed Parameter Tractable (FPT) algorithm for parameter $|\D|$, the number of demands \cite{KKR,RS}. The optimization version of the \textsc{EDPP}, where the number of demands whose routes form a pairwise edge disjoint set is maximized, has an $O(\sqrt{n})$-approximation algorithm \cite{CK}.

 The \textsc{MCRP} is closely related to the Minimum Load Routing Problem \textsc{MLRP}, where the goal is to find a routing which minimizes the load of the network, i.e., the maximum weight of routes sharing an edge. This problem has an approximation factor of $O(\log n/ \log \log n)$ which is matched by an $\Omega(\log n/ \log \log n)$-hardness result \cite{CN}. 

The \textsc{MLRP} restricted to cycles, \textsc{MLRPC}, was first discussed in \cite{CS}, where it was proved to be NP-hard and 2-approximable. This result was improved later in \cite{Khanna}, where the existence of a Polynomial Time Approximation Scheme was proved. It is also known that the problem can be solved in polynomial time when all the weights are the same \cite{Frank}.

In this paper we focus on the MCRP where the input graph is a cycle. This problem, the Minimum Clique Routing Problem on Cycles \textsc{MCRPC}, was first studied in \cite{stefanakos2004routing}. Note that, since any set of routes containing a common edge is a set of pairwise intersecting routes, the optimal value of the \textsc{MCRP} is an upper bound for the optimal value of the \textsc{MLRP}.

In \cite{stefanakos2004routing} the authors construct families of instances with uniform demand weights for which a minimum load routing can result in a clique of size almost twice the optimum.

In the context of the routing in optical networks, instead of minimizing the maximum clique, one is often interested in minimizing the resulting chromatic number of the intersection graph induced by the routing. For instance, in the Routing and Wavelength Assignment Problem  \textsc{RWAP} the task is to find, for each origin-destination pair, a route through the network and a wavelength (or color) so that routes with the same wavelength share no common edges, and the number of required wavelengths is minimized \cite{OB2003}.

\subsection{Our results}
In this work we are interested in the \textsc{MCRPC}, but also in two natural subproblems of it. First, we address the situation where the demands are {disjoint}, in the sense that each node of the cycle may belong to at most one origin-destination pair. Second, we consider the subproblem where the demand weights are {uniform}, i.e., the weights of the origin-destination pairs are all equal.

In Section \ref{sec:np} we show that the \textsc{MCRPC} is NP-complete even if the demands are disjoint. 
In Section \ref{sec:approx} we work with arbitrary weights. We first propose a simple combinatorial 2-approximation algorithm and then we exhibit an LP-based $\frac{3}{2}$-approximation algorithm for the \textsc{MCRPC}, extending to arbitrary weights a similar algorithm proposed in \cite{stefanakos2004routing} for the case of uniform demand weights.

In Section \ref{sec:fpt} we present an FPT algorithm for the particular case of uniform demand weights.

We summarize our results in the following tables.

\bigskip
\begin{center}
\begin{tabular}{|c|c|c|}
\hline
\multirow{2}{*}{} &\multicolumn{2}{|c|}{\large\bf{Arbitrary weights}} \\
\cline{2-3}
& \bf{Disjoint demands} & \bf{General demands} \\
\hline
\multirow{2}{*}{
Complexity} 
& NP-Complete & NP-Complete \\
& (Theorem. \ref{t:nphard}) & (from Theorem \ref{t:nphard}) \\
\hline
\multirow{2}{*}{
Combinatorial approx.} & ${2}$-approx & ${2}$-approx \\
 & (from Theorem. \ref{t:2-approx})& (Theorem. \ref{t:2-approx})\\
\hline
\multirow{2}{*}{
LP-based approx.} & $\frac{3}{2}$-approx & $\frac{3}{2}$-approx \\
 & (from Theorem. \ref{th:lp})& (Theorem. \ref{th:lp})\\
\hline
\end{tabular}

\bigskip

\begin{tabular}{|c|c|c|}
\hline
\multirow{2}{0.75cm}{} &\multicolumn{2}{|c|} {\large\bf{Uniform weights}} \\
\cline{2-3}
& \bf{Disjoint demands} & \bf{General demands} \\
\hline
Complexity & Unknown & Unknown \\
\hline
 \multirow{2}{*}{
Parameterized Complexity} & FPT & FPT \\ 
&(from Theorem \ref{thm:fpt_general}) & (Theorem \ref{thm:fpt_general})\\
\hline
\multirow{2}{*}{
Combinatorial approx.} & ${2}$-approx & ${2}$-approx \\
 & (from Theorem. \ref{t:2-approx})& (Theorem. \ref{t:2-approx})\\
\hline
\multirow{2}{*}{
LP-based approx.} & $\frac{3}{2}$-approx & $\frac{3}{2}$-approx \\
& (from \cite{stefanakos2004routing}) &  (from \cite{stefanakos2004routing})\\
\hline
\end{tabular}
\end{center}

In \cite{alioeuro22} we have presented a preliminary version of some of the results included in this contribution.

\section{Preliminaries}
\label{sec:preliminaries}

In order to state and prove our results, we introduce some definitions and notation.  The graph over which the demands are routed is a cycle on $n$ nodes. We assume that the edges of this cycle have been oriented to obtain the (directed) circuit $(1,2, \ldots,n)$.

For any two distinct nodes $i$ and $j$ in $\{1,2,\ldots,n\}$ we denote by $[i,j]$ the unique directed path from $i$ to $j$ in this circuit and let $\overline{[i,j]} = [j,i]$. Observe that each route in the cycle is given by $[i,j]$, for some $i,j\in \{1,\ldots,n\}$. Moreover, the route $[i,j]$ does not contain the arc $(n,1)$, when $i<j$, and it does, when $j<i$.

A demand is an origin-destination pair  $p=(i,j)$, with $1\leq i<j\leq n$. Recall that we are working in an undirected setting, and therefore $p$ can be routed in two ways: either through the path $p^+=[i,j]$ or through $p^-=[j,i]$. Note that, since $i < j$, $p^-$ contains the arc $(n,1)$, while  $p^+$ does not.  The nodes $i$ and $j$ are called the ends of $p$, and we denote by $e(p) = \{i,j\}$ the set of ends of $p$. 
If two demands have the same ends, we say that these demands are multiples.

The input of \textsc{MCRPC} consists of a set $\D$ of demands and a vector $\w=(\w_p)_{p \in \D}$ of corresponding weights.
A routing $R$ for $\D$ consists in the assignment of one route $R_p \in \{p^+,p^-\}$ to each demand $p \in \D$. Figure~\ref{f:secondexample} depicts an example of an instance with six demands on a cycle with eight nodes.
For a subset of demands $\D'\subseteq \D$ 
we set $\w(\D')=\sum_{p\in \D'}\w_p$. For a subset of routes 
$S$ of a routing $R$ of $\D$, we set $\w(S)=\w(\{p\mid R_p\in S\})$.

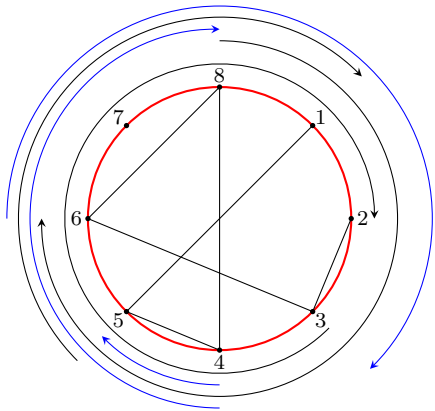
\begin{figure}[H]

\begin{center}
\begin{tikzpicture}[>=stealth,scale=0.7]
\pgfmathsetmacro{\mmt}{0.22}
\pgfmathsetmacro{\mmr}{2.5}

\path[draw,red,thick,
%postaction={decorate,
         %decoration={markings,
         %mark=between positions 0.01 and 1.01 step 1/8 with {\arrow[blue,scale=1.2]{<}; }}}
         ]
         (0:\mmr) arc (0:360:\mmr) -- cycle;
     \foreach \i in {8,...,1} {
    \coordinate (N\i) at (-\i*360/8+90:\mmr cm);
    \fill[black] (N\i) circle (0.05 cm);
    \draw (-\i*360/8+90:\mmr+\mmt) node{{\scriptsize $\i$}};
  }
\draw (N2) -- (N3);
\draw (N4) -- (N5);
\draw (N8) -- (N6);
\draw (N8) -- (N4);
\draw (N1) -- (N5);
\draw (N6) -- (N3);
\draw[->] (0,0) ++ (-45:\mmr+2*\mmt) arc (-45:-360:\mmr+2*\mmt);
\draw[blue,->] (0,0) ++ (-90:\mmr+3*\mmt) arc (-90:-135:\mmr+3*\mmt);
\draw[->] (0,0) ++ (90:\mmr+4*\mmt) arc (90:-180:\mmr+4*\mmt);
\draw[blue,->] (0,0) ++ (-90:\mmr+5*\mmt) arc (-90:-270:\mmr+5*\mmt);
\draw[->] (0,0) ++ (-135:\mmr+6*\mmt) arc (-135:-315:\mmr+6*\mmt);
\draw[blue,->] (0,0) ++ (-180:\mmr+7*\mmt) arc (-180:-405:\mmr+7*\mmt);
\end{tikzpicture}
\end{center}
\caption{An instance of the MCRPC for the cycle of eight nodes, set of demands $\D=\{(1,5),(2,3),(3,6),(4,5),(4,8),(6,8)\}$ and $\w=\mathbf{1}$. Routing $R=\{[3,2],[4,5],[8,6],[4,8],[5,1],[6,3]\}$ of $\D$ has weight 4.}
\label{f:secondexample}
\end{figure}

Let $H$ be the undirected graph with vertex set given by $\{ \hat{p}^+, \hat{p}^- \, : \, p \in \D \}$, where $\hat{p}^+$ (resp. $\hat{p}^-$) represents the route $p^+$ (resp. $p^-$) with associated weights $\w(\hat{p}^+)=\w(\hat{p}^-)=\w(p)$, and two vertices are adjacent if and only if their corresponding routes share at least one common arc. Loosely speaking we call $H$ the intersection graph of the routes.

Since each route is a path of the cycle $G$, the graph $H$ is a circular-arc graph. Moreover, to each routing $R$ of $\D$ we can associated the subgraph of $H$ induced by the routes of $R$, which we denote by $H_R$.

The graph $H$ has interesting structural properties. 
Since any route $p^-$, with $p\in \D$, contains the arc $(n,1)$, the subgraph induced by the set $\{\hat{p}^-\mid p\in \D\}$ is complete. Similarly, as each route $p^+$ is contained in the path from $1$ to $n$, the subgraph induced by the set $\{\hat{p}^+\mid p\in \D\}$ is an interval graph. Figure \ref{f:firstexampleH} follow from the example given in Figure \ref{f:secondexample}. In order to get a cleaner picture we have drawn the non-edges of $H$ and of $H_R$.

\begin{figure}[H]

\begin{tikzpicture}
{\scriptsize
\node (D) at (1,3) {$\D=\{(1,5),(2,3),(3,6),(4,5),(4,8),(6,8)\}$};

\node (R) at (8.25,3) {$R=\{[3,2],[4,5],[8,6],[4,8],[5,1],[6,3]\}$};

\node (p15) at (-1.5,0) {$[1,5]$};
\node (p23) at (-0.5,0) {$[2,3]$};
\node (p36) at (0.5,0) {$[3,6]$};
\node (p45) at (1.5,0) {$[4,5]$};
\node (p48) at (2.5,0) {$[4,8]$};
\node (p68) at (3.5,0) {$[6,8]$};

\node (p51) at (-1.5,2) {$[5,1]$};
\node (p32) at (-0.5,2) {$[3,2]$};
\node (p63) at (0.5,2) {$[6,3]$};
\node (p54) at (1.5,2) {$[5,4]$};
\node (p84) at (2.5,2) {$[8,4]$};
\node (p86) at (3.5,2) {$[8,6]$};

\draw (p15) -- (p51);
\draw (p23) -- (p51);
\draw (p23) -- (p32);
\draw (p36) -- (p63);
\draw (p45) -- (p51);
\draw (p45) -- (p54);
\draw (p48) -- (p84);
\draw (p68) -- (p86);

\draw (p63) -- (p45);
\draw (p84) -- (p45);
\draw (p84) -- (p68);

\draw (p15) to[out=-70,in=-110] (p68);

\draw (p23) to[out=-60,in=-120] (p68);
\draw (p23) -- (p36);
\draw (p23) to[out=-20,in=-120] (p45);

\draw (p23) to[out=-30,in=-110] (p48);

\draw (p45) to[out=-20,in=-140] (p68);

\draw (p36) to[out=-40,in=-130] (p68);

\node (p45) at (9,0) {$[4,5]$};
\node (p48) at (10,0) {$[4,8]$};

\node (p51) at (6,2) {$[5,1]$};
\node (p32) at (7,2) {$[3,2]$};
\node (p63) at (8,2) {$[6,3]$};
\node (p86) at (11,2) {$[8,6]$};

\draw (p45) -- (p51);
\draw (p63) -- (p45);

}
\end{tikzpicture}
\caption{A set of demands, the complement of its associated graph $H$, a routing of the demands and the complement of its associated subgraph $H_R$. }\label{f:firstexampleH}
\end{figure}
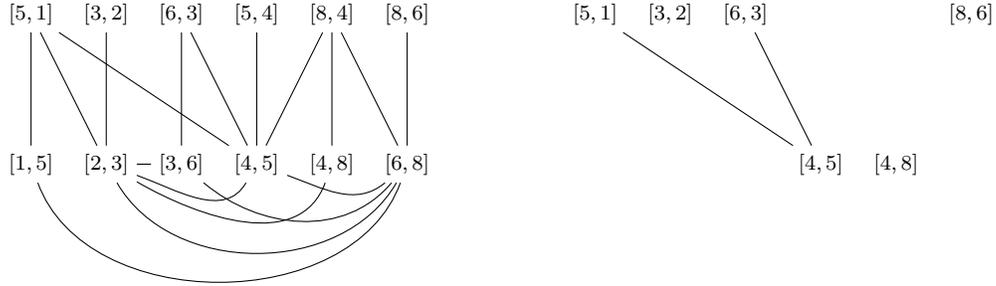

For a given routing $R$, $\omega(H_R,\w)$ denotes the maximum weight of a clique in $H_R$. Since $H_R$ is arc-circular, this value can be computed in time $O(|\D|^3)$, by Hsu's result \cite{Hsu85}.

The  MCRPC consists in determining 
\[c\ell(\D,\w)=\min\{\omega(H_R,\w)\mid R \text{ routing of } \D\}.\] 

For a given demand $p=(i,j) \in \D$, we define 
$V(p^+)$ as the vertices of the directed path $p^+$. $V(p^-)$ is defined in the same way.

\begin{definition}
For given demands $p, q \in \D$ we say that $p$ and $q$ \emph{cross} if $|e(p)\cap V(q^+)|=|e(p)\cap V(q^-)|=1$.
When $p$ and $q$ do not cross we say that they are \emph{parallel}.
\end{definition}

Then, two demands $p$ and $q$ cross if $p$ has exactly one end in each of the possible routes for $q$, i.e., in $q^+$ and $q^-$. Notice that this implies that $q$ also has exactly one end in each of the routes $p^+$ and $p^-$. Moreover, any route of $p$ and any route of $q$ have at least one arc in common. Then, according to this definition, multiple demands are parallel.

When drawing the directed cycle like a circle in the plane, a demand can be represented by a chord between its origin and its destination. In this representation, chords associated with two demands that cross, intersect in the interior of the circle.  In Figure~\ref{f:secondexample}, $(1,5)$, $(4,8)$ and $(3,6)$ are pairwise crossing, while $(1,5)$, $(2,3)$ and $(6,8)$ are pairwise parallel.

\begin{definition}
We say that two demands $p=(i,j)$ and $p'=(i',j')$ are \emph{disjoint} if the nodes $i,j,i',j'$ are all different. 
The \textsc{MCRPC} with disjoint demands is the subproblem where every two different demands of $\D$  are disjoint. 
\end{definition}

\section{NP-completeness with disjoint demands}\label{sec:np}

\begin{theorem}\label{t:nphard} 
The \textsc{MCRPC} with disjoint demands is NP-complete.
\end{theorem}

\begin{proof}
The decision problem associated with the \textsc{MCRPC} is in NP since, given a routing, one can check in polynomial time whether its maximum (weighted) clique 
is upper bounded by $k$ \cite{Hsu85}.

We prove the completeness by a polynomial reduction from the \textsc{Partition} Problem \cite{GJ}.
Let $S = \{x_2,x_{2+1}\ldots, x_{2+(r-1)}\}$ be a multiset of $r$ positive integers, an instance of the \textsc{Partition} Problem (in the reduction we are going to see the convenience of starting the indices from 2). In the \textsc{Partition} Problem we have to decide if there is $T \subseteq S$ such that
$$
\sum_{x \in T} x = \sum_{x \in S \setminus T} x.
$$

We may assume, w.l.o.g., that $M=\sum_{i=2}^{r+1} x_i$ is an even number.
In fact, the \textsc{Partition} Problem remains NP-complete with $M$ even because one could multiply the numbers of the multiset by two, and then this new, modified problem is equivalent to the general one. 

The reduction is as follows. First, the set of nodes in the cycle is $V = \{1,\ldots,2r+4\}$. The set of demands  $\D$ is composed by $r+2$ demands divided into two groups (see Figure \ref{ejemplo}). 

The demands of the first group are the following two: $p_{1}=(1,2)$ and $p_{r+2}=(1+(r+2),2+(r+2))$. We call these two demands {\it{poles}}.
The weights of the poles are given by
$\w(p_{1})=\w(p_{r+2})=M$. 

In the second group of $r$ demands we code the multiset $S$
given by the \textsc{Partition} Problem. More precisely,
the demands are (for $0<i<r-1$):

\begin{align*} 
p_2 &= (3+0,(3+0)+(r+2)) \\ 
\vdots & \\
p_{2+i} &= (3+i,(3+i)+(r+2)) \\ 
\vdots & \\
p_{2+(r-1)} &=(3+(r-1),3+(r-1)+(r+2)) 
\end{align*}

The total weight assigned to these $r$ demands is $M$, which is divided according to the values given by the instance
of the \textsc{Partition} Problem.  More precisely, $\w(p_{i})=x_i$ for all $2\leq i\leq 2+(r -1)$.

 For any routing $R$, recalling that
$R_p$ is the route assigned to $p$, let
 \[\D_+(R) = \{ p \in \D : R_p=p^+ \}\]
and \[\D_-(R) = \{ p \in \D : R_p=p^- \}.\]

Note that the demands in $\D_+(R)$ are those that are routed in such a way that they contain the pole $p_{r+2}$ while avoiding the pole $p_{1}$. The routing $R$ does the opposite on the demands in 
$\D_-(R)$. 

Also, in an optimum routing $R^*$, the two maximum cliques are ``reached" in the poles of the associated circular-arc graph $H$.
In fact, in the optimal routing, both poles $p_{1}$ and $p_{r+2}$ are in  $\D_+(R^*)$ and the intersection of the arcs associated with the two maximum cliques are the unitary 
length arcs $[1,2]$ and $[r+3,r+4]$. Finally, the weights $W$ and $W'$ of the two cliques
are such that $M < W,W' < 2M$. Also: $W+W'=3M$

Let $k=\frac{3}{2}M$.
Now we are going to show that there exists a multiset 
$T \subseteq S$ such that $\sum_{x \in T} x = \sum_{x \in S \setminus T} x = M/2$ if and only if there exists a routing $R$ in $H$ such that $\omega(H_R,\w) = k$.
Suppose that there is some multiset $T$ satisfying $\sum_{x \in T} x = \sum_{x \in S \setminus T} x = M/2$.
Then, the routing
$R$ given by  $R_{p_i} = p_i^+$, for all $i$ such that $x_i \in T$, and $R_{p_i} = p_i^-$ otherwise, together with $R_{p_{1}} = p_{1}^+$ and $R_{p_{r+2}} = p_{r+2}^+$
gives $\omega(H_R,\w)= k$.

Conversely, let $R$ be a routing in $H$ such that $\omega(H_R,\w) = k$. It follows that $\w(\D_+(R)) = \w(\D_-(R)) = k$. This comes from the fact that the sum of the two polar cliques is $W+W'=3M$.
It finally follows that the multiset 
$T = \{x_i : p_i \in \D_-(R) \setminus \{p_{1}\}\}$ must satisfy that
$\sum_{x \in T} x = \sum_{x \in S \setminus T} x = M/2$.

\end{proof}

\begin{remark}
It is worth noting that the previous  polynomial reduction 
can also be constructed from the \textsc{Partition} Problem to
the general \textsc{MCRPC} but where the cycle, instead of being of arbitrary size, has length 6 
(i.e, we change the subproblem of disjoint demands to the subproblem
of fixed-length cycle). The idea is very simple: instead of using $r$ crossing demands (as the $r=4$ crossing demands of Figure \ref{ejemplo}), we use $r$ multiple demands (i.e, all of them sharing the same origin and the same destination). We therefore conclude that  the subproblem
of fixed-length cycle is also NP-complete.
\end{remark}

\begin{figure}[h] 
\begin{center}
\begin{tikzpicture}[>=stealth,scale=0.7]
\pgfmathsetmacro{\mmt}{0.22}
\pgfmathsetmacro{\mmr}{2.5}

\path[draw,red,thick,
%postaction={decorate,
         %decoration={markings,
         %mark=between positions 0.01 and 1.01 step 1/8 with {\arrow[blue,scale=1.2]{<}; }}}
         ]
         (0:\mmr) arc (0:360:\mmr) -- cycle;
     \foreach \i in {12,...,1} {
    \coordinate (N\i) at (-\i*360/12-45:\mmr cm);
    \fill[black] (N\i) circle (0.05 cm);
    \draw (-\i*360/12-45:\mmr+\mmt) node{{\scriptsize $\i$}};
  }
\draw[blue] (N3) -- (N9); 
\draw[blue] (N4) -- (N10); 
\draw[blue] (N5) -- (N11); 
\draw[blue] (N6) -- (N12); 
\draw (N1) -- (N2); 
\draw (N7) -- (N8); 
\node at (0,2) {{\scriptsize $10$}};
\node at (0,-2) {{\scriptsize $10$}};
\node at (0.8,1.3) {{\scriptsize $1$}};
\node at (1.4,0.7) {{\scriptsize $2$}};
\node at (1.8,-0.2) {{\scriptsize $3$}};
\node at (1.4,-1.0) {{\scriptsize $4$}};

\end{tikzpicture}
\caption{A reduction from the \textsc{Partition} Problem to the \textsc{MCRPC} for $r=|S|=4$, where $S=\{x_2,x_3,x_4,x_5\}=\{1,2,3,4\}$. In the reduction we have $M=x_2+x_3+x_4+x_5=10$. Also, $p_1=(1,2)$, $p_2=(3,9)$, $p_3=(4,10)$, $p_4=(5,11)$, $p_5=(6,12)$ and $p_6=(7,8)$. Finally, $\w(p_1)=10$, $\w(p_2)=1$, $\w(p_3)=2$, $\w(p_4)=3$, $\w(p_5)=4$
and $\w(p_6)=10$.}
\label{ejemplo}
\end{center}
\end{figure}
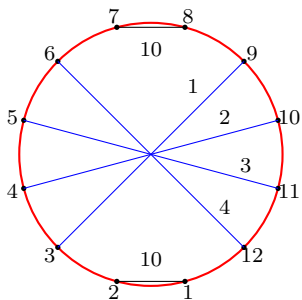

\section{Approximation algorithms}\label{sec:approx}

We first present a simple polynomial time combinatorial algorithm for the \textsc{MCRPC}. Given an instance $(\D,\w)$, it finds a routing $S$ {of $\D$} such that $$\omega(H_S,\w)\leq 2\omega(H_R,\w),$$
for every routing $R$ of $\D$.
Thus, this algorithm is a 2-approximation algorithm for the  \textsc{MCRPC}.

\begin{theorem}\label{t:2-approx}
The  \textsc{MCRPC}  has a 2-approximation algorithm that runs in time $O(|\D|^3)$.
\end{theorem}
\begin{proof}
The algorithm 
builds 
a routing $S$ for the set of 
demands 
$\D$ as follows: the demand $p$ in $\D$ is routed through $p^+$ if $|V(p^+)|\leq |V(p^-)|$ and through $p^-$, otherwise. Then, by using Hsu's algorithm \cite{Hsu85} it solves the weighted clique problem in the circular-arc graph $H_S$ and output a clique $Q$ achieving the maximum. This can be done in time $O(|\D|^3)$.

Let $\D^Q$ be the set of demands associated to the routes in $Q$. 
In the example of Figure 3 the set of demands $\D$ is given by 
$$\D=\{(1,2), (3,9), (4,10), (5,11), (6,12), (7,8)\}.$$
Then, $S_{(i,j)}=(i,j)^+=[i,j]$, for each $(i,j)\in \D$, since 
$$|V([i,j])|=j-i+1\leq |V([j,i])|=12-(j-i)+1,$$ as
$j-i\leq 6$.
Hence, $H_S$ is the disjoint union of a clique and an isolated node. In this case, the maximal clique $Q$ has 5 routes and $\D^Q$ is the set of demands  $\D^Q=\{(3,9), (4,10), (5,11), (6,12), (7,8)\}$.

Thus, 
$\omega(H_S,\w)=\w(\D^Q)$ 
In order to prove that we have a 2-approximation we show that 
$\w(\D^Q)\leq 2\omega(H_R,\w)$, for each routing $R$ of $\D$.

Let $a$ be an arc that appears the most in the routes of $Q$. 
Let $b=(k,l)$ be an arc such that the lengths of the routes $P_a=[j,k]$ and $P_b=[l,i]$ differ by at most one.
(it is unique when the cycle has an even number of nodes).

We prove that each demand in $\D^Q$ has one end in $P_a$ and one end in $P_b$. This is immediate for a such demand $p$ for which either $a$ or $b$ belongs to $A(S_p)$. So let us assume that neither $a$ nor $b$ belong to $A(S_p)$. Then, by the definition of $P_a$ and $P_b$, $S_p$ is contained either in $P_a$ or in $P_b$. Let $c$ be the arc of $S_p$ closest to $a$. Then, as $Q$ is a clique of $H_S$, each route in $Q$ that contains $a$ must also contain $c$. Then, we get the contradiction that $c$ appears in more routes of $Q$ than $a$.

Let $R$ be any routing of $\D$ and let $p\in \D^Q$. 
Since  each demand in $\D^Q$ has one end in $P_a$ and one end in $P_b$, 
the route $R_p$ contains either $a$ or $b$. Let $R^a$ and $R^b$ be the set of routes defined by $R$ in $\D^Q$ containing $a$ and $b$, respectively. Then, $\w(\D^Q)=\w(R^a)+\w(R^b).$

Both sets $R^a$ and $R^b$ form a clique in $H_R$ and then 
the weight of a maximum clique of $H_R$ is at least $\max\{\w(R^a),\w(R^b)\}$. We conclude that for each routing $R$, $\w(\D^Q)=\w(R^a)+\w(R^b)\leq 2\omega(H_R,\w)$.
\end{proof}

Following ideas from  \cite{stefanakos2004routing}, we now present an LP-based polynomial time $3/2$-approximation algorithm for the \textsc{MCRPC}. Although we get a better approximation ratio than in Theorem \ref{t:2-approx}, the algorithm can be slower than the combinatorial one presented therein.

The algorithm proposed in \cite{stefanakos2004routing} consists in rounding a solution of a relaxation of an integer linear programming formulation for  \textsc{MCRPC} with uniform weights. Since the formulation has an exponential number of inequalities, the polynomiality of the algorithm depends on the ability to efficiently solve the corresponding separation problem. This latter problem consists in solving a clique problem in a circular-arc graph which is known to be polynomially solvable \cite{Hsu85}.

The corresponding integer linear programming formulation for  \textsc{MCRPC} with arbitrary weights also  has an exponential number of inequalities. Moreover, 
the separation problem consists in solving a weighted clique problem in a circular-arc graph. Luckily, this problem is also solvable in polynomial time \cite{Hsu85}.
As a result, we can extend the algorithm for uniform weights to the general situation, as it is proved in the next theorem.

\begin{theorem}\label{th:lp}
The \textsc{MCRPC} has a polynomial $\frac{3}{2}$-approximation algorithm.
\end{theorem}

\begin{proof}

Following the notation of \cite{stefanakos2004routing}, we introduce for each demand $p \in \D$ two binary variables $x_p$ and $y_p$, with $x_p + y_p = 1$, to indicate whether $p$ is routed over $p^+$ ($x_p = 1$) or over $p^-$ ($y_p = 1$).

Remember that $H$ was defined as
the circular-arc graph obtained as the intersection graph of \emph{all} possible demand routes, i.e., the intersection graph of the set of paths $\{ p^+, p^- \, : \, p \in \D \}$, and let $\CH$ be the set of all cliques of $H$. Then the problem of finding a minimum-weight clique routing 
for a set of demands $\D$
can be formulated as:

\[
\text{(MCR)}\left\{
\begin{aligned}
    &\min K \\
    &\text{s.t.}\\
    & \sum_{p^+ \in C} w_p x_p + \sum_{p^- \in C} w_p y_p \leq K,\quad \forall C \in \CH, \\
    & x_p + y_p = 1, \quad \forall p \in \D,\\
    & x_p, y_p \in \{0,1\}, \quad \forall p \in \D.
\end{aligned}
\right.
\]

Due to the equivalence between optimization and separation, the linear relaxation of this program can be solved in polynomial time. Indeed, given a fractional solution $(x^*, y^*)$ the separation problem can be reduced to finding a maximum-weight clique on the circular-arc graph 
$H$ with respect to the weights $w_p x^*_p$ for $p^+$ and $w_p y^*_p$ for $p^-$, for all $p \in \D$.

Moreover,  
we show that a rounding scheme proposed in \cite{stefanakos2004routing} can be generalized to this case with arbitrary demands, as follows.

Assume in a fractional solution $(x^*, y^*)$ two parallel demands $p$ and $q$ have been routed as indicated in Figure~\ref{fig_parallel_rounding}(a), where the amount of each demand that has been routed along the two possible routes is indicated by the value shown on the corresponding side of the circle chord. For instance, for demand $p$, $x^*_p w_p$ units have been routed along $p^+$, whereas the remaining $y^*_p w_p$ units have been routed along $p^-$. Furthermore, assume that $y^*_p w_p \leq x_q^*w_q$, otherwise the construction described below may be modified accordingly.

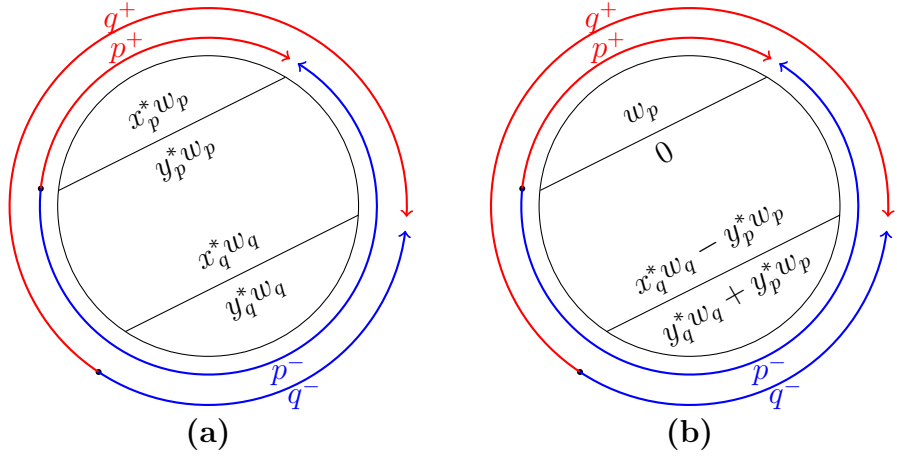
\begin{figure}
    \centering
    \begin{tikzpicture}[scale=0.8]
        \draw (3,4) circle (2.5cm);
        \draw (11,4) circle (2.5cm);
        \draw (0.52,4.26) -- (4.28,6.14);
        \fill (0.22,4.29) circle (0.5mm);
        \draw[red,->,thick] (0.22,4.29) arc (174: 61 : 2.8cm);
        \draw[blue,->,thick] (0.22,4.29) arc (174: 418 : 2.8cm);
        \draw[red] (1.66,6.63) node {$p^+$};
        \draw[blue] (4.34,1.27) node {$p^-$};
        \draw (2.4,5.2) node[rotate=30, anchor= south] {$x^*_p w_p$};
        \draw (2.4,5.2) node[rotate=30, anchor= north] {$y^*_p w_p$};
        \draw (1.62,1.91) -- (5.5,3.85);
        \fill (1.18,1.24) circle (0.5mm);
        \draw[red,->,thick] (1.18,1.24) arc (236.5: -3 : 3.3cm);
        \draw[blue,->,thick] (1.18,1.24) arc (236.5: 353 : 3.3cm);
        \draw[red] (1.54,7.13) node {$q^+$};
        \draw[blue] (4.6,0.87) node {$q^-$};
        \draw (3.56,2.88) node[rotate=30, anchor= south] {$x^*_q w_q$};
        \draw (3.56,2.88) node[rotate=30, anchor= north] {$y^*_q w_q$};
        \draw (8.52,4.26) -- (12.28,6.14);
        \fill (8.22,4.29) circle (0.5mm);
        \draw[red,->,thick] (8.22,4.29) arc (174: 61 : 2.8cm);
        \draw[blue,->,thick] (8.22,4.29) arc (174: 418 : 2.8cm);
        \draw[red] (9.66,6.63) node {$p^+$};
        \draw[blue] (12.34,1.27) node {$p^-$};
        \draw (10.4,5.2) node[rotate=30, anchor= south] {$w_p$};
        \draw (10.4,5.2) node[rotate=30, anchor= north] {$0$};
        \draw (9.62,1.91) -- (13.5,3.85);
        \fill (9.18,1.24) circle (0.5mm);
        \draw[red,->,thick] (9.18,1.24) arc (236.5: -3 : 3.3cm);
        \draw[blue,->,thick] (9.18,1.24) arc (236.5: 353 : 3.3cm);
        \draw[red] (9.54,7.13) node {$q^+$};
        \draw[blue] (12.6,0.87) node {$q^-$};
        \draw (11.56,2.88) node[rotate=30, anchor= south] {$x^*_q w_q - y^*_p w_p$};
        \draw (11.56,2.88) node[rotate=30, anchor= north] {$y^*_q w_q + y^*_p w_p$};
        \draw (3,0.2) node {\textbf{(a)}};
        \draw (11,0.2) node {\textbf{(b)}};
    \end{tikzpicture}
    \caption{Re-routing scheme for parallel fractional demands: (a) Two parallel demands with fractional routings; (b) After re-routing, $x_p=1$ and $y_p = 0$.}
    \label{fig_parallel_rounding}
    
\end{figure}

In this case, we may re-route $y_p^* w_p$ units of demand $p$ from $p^-$ to $p^+$, and the same amount of units of demand $q$ from $q^+$ to $q^-$, obtaining the new routing depicted in Figure~\ref{fig_parallel_rounding}(b). In this solution, the routing for demand $p$ is no longer fractional, as $x_p=1$ and $y_p=0$. Moreover, the load on the routes $p^+$ and $q^-$ has increased by $y_p^* w_p$ units, while the load on the routes $p^-$ and $q^+$ has decreased by the same amount. Since the route $p^+$ (resp. $q^-$) is contained in the route $q^+$ (resp. $p^-$), it follows that 
any clique of the circular-arc graph 
$H$ containing $p^+$ (resp. $q^-$) must also contain  $q^+$ (resp. $p^-$). Hence, the weight of the cliques in $H$ does not increase after this re-routing. 

As a consequence from the last observation, we may assume that all demands corresponding to fractional variables in the optimal solution of the linear relaxation of (MCR) are pairwise crossing. 

Let $I$ be an instance of (MCR) and $Q \subset \D$ be the set of pairwise crossing demands corresponding to the variables with fractional values in the optimal solution $(x^*,y^*)$ of the linear relaxation of (MCR). Observe that $H$ has a clique of weight $w(Q)$ for any (integral) routing. Thus, if  $\text{OPT}(I)$ is the optimal value of (MCR) for the instance $I$, $w(Q) \leq \text{OPT}(I)$. 

Now consider the following rounding scheme: for any $p \in Q$,  all demand weight $w_p$ is routed along $p^+$ if $y^*_p \leq x^*_p$, or along $p^-$, otherwise. 
Hence, the weight of any clique containing either $p^+$ or $p^-$ increases at most by $\frac{1}{2}w_p$. In total, this
rounding does not increase the value $\text{OPT}_f(I)$ of the fractional solution by more than $\frac{1}{2}w(Q)$. Since $\text{OPT}_f(I)$ is a lower bound on $\text{OPT}(I)$ and $w(Q) \leq \text{OPT}(I)$, a feasible solution for (MCR) is obtained with value less than or equal to $\text{OPT}_f(I) + \frac{1}{2}w(Q) \leq 
\text{OPT}(I) + \frac{1}{2}\text{OPT}(I) =
\frac{3}{2}\text{OPT}(I)$.
\end{proof}

\section{A  Fixed Parameter Tractable  algorithm for uniform weights}\label{sec:fpt}

In this section we assume that all weights are uniform, and w.l.o.g that all are equal to 1.
In this case the complexity of the \textsc{MCRPC}  is not known. In fact, the subproblem we used to show that \textsc{MCRPC} is NP-hard for arbitrary weights, can be solved in polynomial time when the weights are all the same. We now present a structural property that we use later to show that \textsc{MCRPC} admits a  Fixed Parameter Tractable (FPT) algorithm.

Given a routing $R$, two parallel demands $p$ and $q$  \emph{collide} in $R$ 
if $A(R_p)\cup A(R_q)$ is the set of arcs of the cycle and $A(R_p)\cap A(R_q)\neq \emptyset$.  In this case, we call the pair $(p,q)$ a \emph{collision} in $R$.
Notice that two demands sharing the same endpoints (multiple demands) do not collide. 

The existence of our  Fixed Parameter Tractable algorithm for the \textsc{MCRPC} is based on the following structural property of its optimal solutions.

\begin{theorem}\label{t:structoptAB}
The \textsc{MCRPC} for uniform weights has an optimal solution without collisions. 
\end{theorem}
\begin{proof}
By contradiction, let us assume that every optimal solution has collisions and take an optimal solution $R$ with the minimum number of them. Let $(p,q)$ be a collision in $R$ such that $|A(R_p)|+|A(R_q)|$ is as large as possible. 

Consider the new solution $R'$ obtained from $R$ by replacing $R_p$ and $R_q$ by $\overline{R_p}$ and $\overline{R_q}$, respectively.
Then, in  $H_{R'}$ the vertices $\hat{p}$ and $\hat{q}$ associated with the routes $R'_p$ and $R'_q$ are not adjacent. Moreover, any
neighbor of the vertex $\hat{p}$ in  $H_{R'}$ is also a neighbor of the vertex associated to $R_q$ in  $H_{R}$ because $R'_p\subseteq R_q$; similarly, any neighbor of the vertex $\hat{q}$ in $H_{R'}$ is also a neighbor of the vertex associated to $R_p$ in $H_{R}$.

Let $Q$ be a set of pairwise adjacent vertices in  $H_{R'}$ such that $|Q|=\omega(H_{R'},\mathbf{1})$. We prove that $|Q|\leq \omega(H_R,\mathbf{1})$. 

If $Q$ neither contain $\hat{p}$ nor $\hat{q}$, then the vertices in $Q$ are also pairwise adjacent in $H_{R}$. Whence,  $|Q|\leq \omega(H_R,\mathbf{1})$.
Otherwise, since $\hat{p}$ and $\hat{q}$ are not adjacent in $H_{R'}$, we can assume that exactly one of them, let us say $\hat{p}$, belongs to $Q$. Hence, the vertices in $Q\setminus \{\hat{p}\}$ are adjacent to the vertex associated to $R_q$ in $H_{R}$ thus defining a clique in $H_{R}$ of size $|Q|$. This shows that $|Q|\leq \omega(H_R, \mathbf{1})$.
Therefore, $\omega(H_{R'},\mathbf{1})\leq \omega(H_R, \mathbf{1})$.

In order to get a contradiction we observe that if $o \in \D$ with $e(o)\neq e(p)$ and $e(o)\subseteq V(\overline{R_p})$ then, $R_o \subseteq \overline{R_p}$. Otherwise, $R_p \subseteq R_o$ and $(o,q)$ would be a collision, because $e(q)\subseteq R_p$. Since $|A(R_p)| < |A(R_o)|$, we would get a contradiction. Similarly, if $e(o)\neq e(q)$ and $e(o)\subseteq V(\overline{R_q})$, then $R_o \subseteq \overline{R_q}$. This shows that neither $p$ nor $q$ are involved in any collision in $R'$.
As any collision in $R'$ not involving neither $p$ nor $q$ is a collision in $R$ we conclude that the number of collisions in $R'$ is smaller than the ones in $R$. 
\end{proof}

Even though we do not known how to find an optimal routing for \emph{all} the demands (in the general case), we can find an optimal routing when they are all parallel to some fixed demand. In order to present this result more precisely, we need some definitions. 

For a demand $p\in \D$,  $\D(p)$ denotes the multiset of all demands $q$ parallel to $p$ such that $e(q)\neq e(p)$. Let $S^p$ be the routing of $\D(p)$ given by:
$S^p_q\subseteq p^+$, when $e(q)\subseteq V(p^+)$, and $S^p_q\subseteq p^-$, when $e(q)\subseteq V(p^-)$. 
Also, consider  $c_p^+=|\{q: S^p_q\subseteq p^+,q\neq p\}|$  and  $c_p^-= |\{q: S^p_q\subseteq p^-,q\neq p\}|$.


\begin{prop}\label{p:isolatedAB}
Let $p\in \D$ be such that $\D=\D(p)\cup \{p\}$. If  $c_p^+\geq c_p^-$, there exists an optimal routing $R$ without collisions such that $R_p=p^-$.
\end{prop}

\begin{proof} From Theorem \ref{t:structoptAB} we know that there is an optimal solution without collisions.
Suppose that $R'$ is an optimal routing without collisions such that $R'_p=p^+$.
{We know that for} each $q$, with $S^p_q\subseteq p^+$, {we have that } $R'_q=q^+$, in order  to avoid collisions with $R'_p$. Then, $\omega(H_{R'},\mathbf{1})\geq c_p^+ +1$. 

Let $R$ be the extension of $S^p$ to a routing of $\D$ given by $R_p=p^-$ and 
$R_q=S^p_q$, for $q\neq p$. Then, 
we get that  $\omega(H_R,\mathbf{1})=\max\{c_p^+,c_p^- +1\}\leq c_p^+ +1$, since $c_p^+\geq c_p^-$, by hypothesis. This shows that $\omega(H_R,\mathbf{1})\leq \omega(H_{R'},\mathbf{1})$. Therefore,  since $S^p$ has no collisions, 
$R$ is an optimal routing  without collisions.
\end{proof}

A demand $p$ is \emph{critical} for a routing $R$ if  for all $q\in \D(p)$, $R_q=S^p_q$.

It is easy to see that if $p$ is such that $R_p$ has the largest number of edges in  a routing $R$ without collisions, 
then $p$ is critical for $R$. Notice that the number of routings for which a given demand $p$ is critical is at most $2^{|\D\setminus \D(p)|}$. {Thus, if $k=\max\{|\D\setminus\D(p)|: p\in \D\}$, we can derive an FPT algorithm parameterized by $k$ since this number provides a bound on the possible routings having some critical demand in $\D$.}
Notice also that when $\D$ has no multiple demands, for each demand $p\in \D$, the value $|\D\setminus \D(p)|$ is a lower bound for the degree of $p$ in the graph $H_R$, for any routing $R$. Hence, in this case, $k$ is a lower bound for the maximum degree of $H_R$.

\begin{theorem}\label{thm:fpt_general}
 The \textsc{MCRPC} has an FPT algorithm with time complexity $O(2^k|\D|^4)$.
\end{theorem}
\begin{proof}
We define an  algorithm that computes, for each demand $p$, a routing called $R^p$ which minimizes $\omega(H_{R},\mathbf{1})$, over all routings $R$ of $\D$ for which $p$ is critical for $R$. 
The output of the algorithm is a routing $R^*$ that minimizes $\omega(H_{R^p},\mathbf{1})$ over all $p\in \D$.

The algorithm is correct since each routing without collisions has a critical demand and then, an optimal solution to our problem without collisions must be $R^p$, for some $p\in\D$. 

For a given routing $R$, computing $\omega(H_{R},\mathbf{1})$ takes time $O(|\D|^3)$. Then, to determine $R^p$ takes time $O(2^k|\D|^3)$, since the set of routings for which $p$ is critical has at most $2^k$ elements. Therefore, the time complexity of this algorithm is $O(2^k|\D|^4)$.
\end{proof}

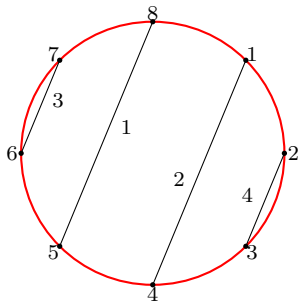
\begin{figure}\label{f:optcoll}
\centering
\begin{tikzpicture}[>=stealth,scale=0.7]
\pgfmathsetmacro{\mmt}{0.17}
\pgfmathsetmacro{\mmr}{2.5}

\path[draw,red,thick]
         
         (0:\mmr) arc (0:360:\mmr) -- cycle;
     \foreach \i in {8,...,1} {
    \coordinate (N\i) at (-\i*360/8+90:\mmr cm);
    \fill[black] (N\i) circle (0.05 cm);
    \draw (-\i*360/8+90:\mmr+\mmt) node{{\scriptsize $\i$}};
  }
\draw (N5) -- (N8);
\draw (N6) -- (N7);
\draw (N3) -- (N2);
\draw (N1) -- (N4);

\draw (-0.5,0.5) node{{\scriptsize $1$}};
\draw (0.5,-0.5) node{{\scriptsize $2$}};
\draw (1.8,-0.8) node{{\scriptsize $4$}};
\draw (-1.8,1) node{{\scriptsize $3$}};

\end{tikzpicture}
\caption{The optimal solution to the instance of the figure has value 5 and any feasible solution without collisions has value at least 6.}
\end{figure}

{It is worth to notice that the} property allowing to show that the FPT algorithm works for uniform weights cannot be extended to the case of arbitrary weights, as the  example of Figure \ref{f:optcoll} shows.

\section{Conclusions}

We have proved that for arbitrary weights \textsc{MCRPC} is NP-hard and it has a 3/2-approximation algorithm. 

We also have proposed an FPT algorithm for the case when all weights are equal to one. 
Unfortunately, despite a great amount of effort, we have been unable to determine the complexity of \textsc{MCRPC} in this case. However, for some special cases (those appearing in the proof of the NP-hardness of \textsc{MCRPC}) we know that the unweighted case has a polynomial time algorithm whereas the weighted case is still NP-hard. 

Surprisingly, even though we do not know how to get an optimal solution when all the weights are the same, we do know how to route each \emph{extreme} demand, i.e., a demand $p$ such that either $p^+$ or $p^-$ does not contain the two ends of another demand.
In fact, one can see that if $p^+$ does not contain the two ends of another demand, then the vertex representing $p^-$ in $H$ is universal. Thus, if $R$ routes $p$ through $p^-$, then the routing $R'$ obtained from $R$ by just changing the route of  $p$ from $p^-$ to $p^+$ defines a graph $H_{R'}$ whose maximum clique has at most the size of a maximum clique of $H_{R}$. 

Notice that this idea allows us to reduce the original problem to a smaller one consisting only of demands which are not extreme. However, this new problem is a weighted problem. Even worse, the remaining demands $q$ might have a different weights for $q^+$ and $q^-$. 
\section*{Acknowledgements}

We thank Annegret Wagler for earlier discussions on the Routing and Spectrum Assignment (RSA) problem on trees and cycles, which motivated our work on this topic.

This work was partially supported by: STIC-AMSUD 22-STIC-08, EPN PIGR-19-11 (L.M. Torres), AMSUD210003 (M. Matamala and I. Rapaport), BASAL ANID-Chile (M. Matamala and I. Rapaport), FONDECYT 1220142 (I.Rapaport), PICT 2020-03032 Serie A ANPCyT and  PID UNR  80 02 02 10 20 00 06 UR (M. Escalante and P.Tolomei).

\bibliographystyle{plain}

\end{document}